\documentclass[a4paper,12pt]{article}

\usepackage[english]{babel}
\usepackage{authblk}
\usepackage{dsfont}
\usepackage{amsfonts}
\usepackage{mathrsfs}
\usepackage{amssymb}        
\usepackage{amsmath}
\usepackage{graphicx,caption}
\usepackage{float}
\usepackage[a4paper]{geometry} 
\usepackage{verbatim}
\usepackage{pstricks}
\usepackage{amsthm}
\usepackage{mathrsfs}
\usepackage{enumitem}

\newtheorem{theorem}{Theorem}
\newtheorem{proposition}{Proposition}

\newtheorem{lemma}{Lemma}

\newtheorem{remark}{Remark}

\numberwithin{equation}{section}

\newcommand{\C}{\mathbb{C}}
\newcommand{\R}{\mathbb{R}}

\newcommand{\Z}{\mathbb{Z}}
\newcommand{\N}{\mathbb{N}}

\newcommand{\dd}{\mathrm{d}}

\newcommand{\ed}{\mathrm{e}}

\newcommand{\bx}{{\mathbf{x}}}
\newcommand{\by}{{\mathbf{y}}}
\newcommand{\bsigma}{{\boldsymbol{\sigma}}}
\newcommand{\btau}{{\boldsymbol{\tau}}}

\newcommand{\Const}{\mathrm{C}}
\newcommand{\const}{c}

\newcommand{\Proba}{\mathsf P}
\newcommand{\imag}{\mathrm i}
\newcommand{\randomt}{{\mathsf t}}

\newcommand{\har}{{\mathrm{har}}}
\newcommand{\an}{{\mathrm{an}}}

\begin{document}

\title{Sub-Power Law Decay of the Wave Packet Maximum in Disordered Anharmonic Chains}

\author[1]{Wojciech De Roeck}
\author[2]{Lydia Giacomin}
\author[1]{Amirali Hannani}
\author[3]{François Huveneers}
\affil[1]{Institute for Theoretical Physics, KU Leuven, 3000 Leuven, Belgium}
\affil[2]{Ceremade, UMR-CNRS 7534, Université Paris Dauphine, PSL Research
University, 75775 Paris, France}
\affil[3]{Department of Mathematics, King’s College London, Strand, London WC2R 2LS, United Kingdom}

\date{\today}
\maketitle 

\begin{abstract}
We show that the peak of an initially localized wave packet in one-dimensional nonlinear disordered chains decays more slowly than any power law of time.
The systems under investigation are Klein-Gordon and nonlinear disordered Schrödinger-type chains, characterized by a harmonic onsite disordered potential and quartic nearest-neighbor coupling.
Our results apply in the long-time limit, hold almost surely, and are valid for arbitrary finite energy values.
\end{abstract}

\section{Introduction}

The spreading of wave packets in disordered systems is a fascinating problem that has drawn attention from both mathematicians and physicists. 
In this paper, we explore the behavior of an initially localized wave packet in a one-dimensional disordered lattice, governed by either a nonlinear disordered Klein-Gordon (KG)-type equation or a disordered nonlinear Schrödinger (DNLS)-type equation, considered in the atomic limit where harmonic nearest-neighbor coupling is suppressed.

When all couplings in the lattice are harmonic, making the dynamics linear, it is well known that the wave packet remains localized due to Anderson localization, see \cite{anderson_absence_1958} as well as \cite{gol1977pure,molchanov78,kunz_souillard,carmona} for mathematical results in one dimension.
However, the introduction of anharmonic interactions makes the equations of motion nonlinear and brings new and intriguing dynamics into play. 
The prevailing conjecture is that the wave packet begins to spread, provided its initial energy is sufficiently high. 
Despite numerous numerical \cite{pikovsky_destruction_2008,garcia-mata_delocalization_2009,flach_universal_2009,skokos_delocalization_2009,skokos_spreading_2010,mulansky_scaling_2011,vakulchyk_wave_2019}, 
theoretical \cite{fishman_perturbation_2009,fishman_nonlinear_2012,basko_weak_2011} and analytical \cite{frohlich_localization_1986,benettin_nekhoroshev_1988,poschel_small_1990,bourgain_wang_2007,wang_long_2009,cong_long-time_2021,cong_2024} works, the exact nature and rate of this spreading remain elusive and continue to spark curiosity.
See also \cite{basko_weak_2011,de_roeck_slow_dissipation_2024} for further refererences.

There is a clear tension between the conclusions drawn from numerical studies and those suggested by the theoretical and mathematical analyses cited above. On the one hand, numerical works consistently report that the wave packet spreads as a power law over time, with a dynamical exponent that depends on the non-linearity of the model. A notable exception is \cite{mulansky_scaling_2011}, which examines the exact same model studied in this paper and observes deviations from the power-law regime, aligning with an effective description. Additionally, a few numerical experiments at positive energy density hint indirectly at the possibility of non-power-law spreading of the wave packet \cite{oganesyan_pal_huse_2009,kumar_transport_2020}. 

Theoretical and mathematical studies on the other hand, consistently suggest ultra-slow, non-power-law spreading, a conjecture further supported indirectly by mathematical analyses of out-of-time equilibrium correlators \cite{huveneers_drastic_2013,de_roeck_long_persistence_2023,de_roeck_slow_dissipation_2024}. Despite the considerable progress they bring to the field, existing mathematical works fall short of providing a definitive answer for several reasons: they may apply only to a restricted set of initial conditions \cite{frohlich_localization_1986,poschel_small_1990}, hold only for long transient times \cite{benettin_nekhoroshev_1988,wang_long_2009,cong_long-time_2021,cong_2024}, or rely on simplifying modifications to the model \cite{bourgain_wang_2007}.
The aim of this paper is to overcome all these limitations.

Two common ways of quantifying the spreading of wave packets have been considered. Since both the KG and DNLS chains conserve total energy, expressed as $H = \sum_{x \in \mathbb{Z}} H_x$, it is natural to describe the packet's location based on the distribution of its energy. One method is to compute a ``moment'' of the packet, with the second moment being the most commonly used:  
\[
q_2(t) \; = \; \sum_{x \in \mathbb{Z}} x^2 H_x(t),
\]  
which provides a measure of how far the energy is distributed.  
Another way is to use the inverse participation ratio:  
\[
    r_2(t) \; = \; \left( \sum_{x \in \mathbb{Z}} H_x^2(t) \right)^{-1}, 
\]  
which measures the number of sites over which the energy is concentrated. This last measure is directly related to the maximum value of the wave packet, \( M(t) = \max_{x \in \mathbb{Z}} H_x(t) \), through the inequality:  
\[
    \frac{1}{E_0 M(t)} \;\leq\; r_2(t) \;\leq\; \frac{1}{M(t)^2},
\]  
where $E_0$ is the total energy of the packet.

If the wave packet spreads properly, staying within a reasonable envelope, we expect the quantities $q_2(t)$ and $r_2(t)$ to be related by $q_2(t) \sim r_2(t)^2$. However, if the packet only wanders without spreading, or if part of it remains localized, the second moment $q_2(t)$ may grow on average while the inverse participation ratio $r_2(t)$ does not. Numerical experiments that have investigated this aspect consistently support the first scenario: the wave packet spreads properly, at least over the time scales accessible to numerical simulations;
see e.g.\@ \cite{skokos_spreading_2010}.

From a mathematical perspective, studying $q_2(t)$ or $r_2(t)$ presents very different challenges. In the absence of any explicit dissipative term in the equations of motion, proving that the wave packet spreads rather than simply wanders is extremely difficult. Since spreading is what ensures that the non-linearity becomes effectively weaker over time, obtaining interesting bounds on $q_2(t)$ remains probably out of reach with current methods. In contrast, analyzing the decay of the wave packet’s maximum, $M(t)$, and consequently $r_2(t)$, is far more accessible. Smaller values of $M(t)$ directly correspond to weaker effective non-linearity, making $r_2(t)$ a more tractable quantity for mathematical analysis.

In this work, we demonstrate that the maximum $M(t)$, and consequently the inverse participation ratio $r_2(t)$, decays slower than any inverse power of $t$ for dynamics governed by an onsite quadratic potential and quartic nearest-neighbor interactions. Our analysis applies to both the KG and DNLS chains. 
To make our novel method as clear as possible, we have chosen to present it in a system where its application is most transparent. However, we are confident that this approach extends to a broader class of disordered systems, the exploration of which we leave for future work.

\paragraph{Acknowledgments:}
A.~H.\@ and W.~D.~R.\@ were supported by the FWO-FNRS EOS research project G0H1122N EOS 40007526 CHEQS, the KULeuven Runner-up grant iBOF DOA/20/011, and the internal KULeuven grant C14/21/086.

\section{Model and Result}

Let 
$\ell^2(\Z,\R^2)$ be the Hilbert space of sequences $(q,p)=(q_x,p_x)_{x\in\Z}$ such that $\sum_{x\in\Z} q_x^2 + p_x^2$ is finite, 
and let the Hamiltonian $H:\ell^2(\Z,\R^2)\to\R$ be given by
\begin{equation}\label{eq: Hamiltonian}
    H(q,p) 
    \; = \;
    \sum_{x\in\Z} \frac{p_x^2}{2} + \frac{\omega_x^2 q_x^2}{2} + U(q_x - q_{x+1})
    \; = \; 
    \sum_{x\in\Z} H_x(q,p)
\end{equation}
with 
\[
    U(r) \; = \; \frac{g}4 r^4, \qquad g \ge 0. 
\]
Here $(\omega_x^2)_{x\in\Z}$ is an i.i.d.\@ sequence of random variables with a smooth density supported on the compact interval $[\omega_-^2,\omega_+^2]$, 
where $0 < \omega_-^2 < \omega_+^2 < + \infty$.
By energy conservation, the Hamiltonian dynamics 
\[
    \frac{\dd q}{\dd t} \; = \; p, \qquad
    \frac{\dd p}{\dd t} \; = \; -\nabla_q H
\]
is well defined: $(q(t),p(t))\in\ell^2(\Z,\R^2)$ for all times $t\ge 0$ provided  $(q(0),p(0))\in\ell^2(\Z,\R^2)$.
We note that a smooth, local observable $O$ evolves as 
\[
    \frac{\dd O}{\dd t} \; = \; \lbrace H,O\rbrace
\]
with the Poisson bracket $\{\cdot,\cdot\}$ being defined as 
\[
    \{f,g\} \; = \; \nabla_p f \cdot \nabla_q g - \nabla_q f \cdot \nabla_p g
\]
for smooth observables $f,g$. 

Let us now describe the initial condition $(q(0),p(0))$ considered throughout this paper. 
We say that $I\subset \Z$ is an \emph{interval} if $I = I'\cap\Z$ for some interval $I'\subset\R$. 
We assume that there exists a compact interval $I_0\subset \Z$ containing $0$ as well as a positive number $E_0>0$ such that 
\begin{enumerate}
    \item 
    $(q(0),p(0))$ is supported in $I_0$, 

    \item 
    $\max_{x\in\Z} H_x(q(0),p(0))
      =  H_0(q(0),p(0)) = E_0$.
\end{enumerate}

This paper deals with the long time behavior of 
\[
    M(t) \; = \; \max_{x\in\Z}H_x(t).
\]
We prove
\begin{theorem}\label{th: main theorem}
    There exists a random variable $\randomt^*\ge 1$ with finite moments of all orders such that
    \begin{equation}\label{eq: result main theorem}
        t \; \ge \; \randomt^*
        \qquad \Rightarrow \qquad
        M(t) \; \ge \; \ed^{-2 (\ln t)^{3/4}}.
    \end{equation}
\end{theorem}

\begin{remark}\label{rem: non-optimality}
    The exponent \(3/4\) in this theorem is not claimed to be optimal. What matters is that \( M(t) \) decays slower than any inverse power of \( t \) as \( t \to \infty \).
\end{remark}

\begin{remark}\label{rem: dnls chain}
    The result holds as well for the \emph{Discrete Non-Linear Schrödinger Equation (DNLS)}, where the Hamiltonian in \eqref{eq: Hamiltonian} is replaced by 
    \begin{equation}\label{eq: Hamiltonian dnls}
    H(\psi) \; = \; \sum_{x\in\Z} \omega_x |\psi_x|^2 + (g/4)|\psi_x - \psi_{x+1}|^4
    \end{equation}
    with $\psi = (\psi_x)_{x\in\Z}\in\ell^2(\Z,\C)$. 
    In this case, Hamilton's equations take the form of the Schrödinger equation
    \begin{equation}\label{eq: schrodinger equation}
        \imag \frac{\dd \psi}{\dd t} \; = \; \nabla_{\overline\psi}H
    \end{equation}
    where we view $H$ in \eqref{eq: Hamiltonian dnls} as a function of $\psi$ and $\overline\psi$ using the identity $|z|^2=z\overline z$. 
    See Section~\ref{sec: DNLS}.
\end{remark}


The remainder of this text is devoted to the proof of Theorem~\ref{th: main theorem}. 
We will denote by \emph{constants} positive, deterministic, real numbers that may depend on the  distribution of the disorder, $g$, $I_0$ and $E_0$.
We will often use the letters $\Const$ and $\const$ to denote generic constants, with the understanding that their values may vary from one instance to another.
Usually, we use the letter $\Const$ to stress that the constant needs to be taken large enough, and the letter $\const$ to stress that it needs to be taken small enough.

\section{A Priori Bounds}

We provide two a priori bounds that will be used throughout the remainder of this paper. 

\subsection{Light-Cone-Type Bound}
For $t\ge 0$, let $x(t)$ be the largest point in $\Z$ such that 
\begin{equation}\label{eq: def of x(t)}
    H_{x(t)}(t) \; = \; M(t).
\end{equation}
Given a site $x\in\Z$, we define also the energy current between site $x-1$ and $x$ as 

\begin{equation}\label{eq: jx}
    j_x \; = \; \{H_{x-1},H_x\} \; = \; p_x U'(q_{x-1}-q_x) \; = \; g p_x (q_{x-1}-q_x)^3.
\end{equation}
With this definition, the continuity equation for the energy writes 
\begin{equation}\label{eq: continuity equation}
    \frac{\dd H_x}{\dd t} 
    \; = \; \{H,H_x\} 
    \; = \; \{H_{x-1},H_x\} - \{H_x,H_{x+1}\}
    \; = \; j_x - j_{x+1}.
\end{equation}

We establish a light-cone-type bound that constrains the position of $x(t)$: 
\begin{lemma}\label{lem: light cone}
    There exists a constant $\Const$ such that for any $t\ge 1$, 
    \[
    M(t) |x (t)| \;\le\; \Const t. 
    \]
\end{lemma}

\begin{proof}
    The bound is clearly satisfied as long as \( x(t) \in I_0 \). 
    Therefore, we focus on the case where \( x(t) \notin I_0 \).  
    Specifically, let us consider the case \( x(t) > \max I_0 \); the case \( x(t) < \min I_0 \) can be treated analogously. 
    For $x>\max I_0$, $H_x(0) = 0$ and the continuity equation~\eqref{eq: continuity equation} yields 
    \[
        \int_0^t \dd s\, j_x(s) \; = \; \sum_{x'\ge x} H_{x'}(t).
    \]
    If, in addition, $x\le x(t)$, the right hand side of this equality is lower bounded by $M(t)$. 
    Therefore 
    \[
        M(t) \; \le \; \frac{1}{x(t) - \max I_0}\sum_{\max I_0 < x \le x(t)}\int_0^t \dd s \, j_x(s) .
    \]
    Now we claim that there exists a constant $\Const$ such that
    \begin{equation}\label{eq: bounded current}
        \sum_{x\in \Z}\int_0^t \dd s\, |j_x(s)| \; \le \; \Const t.
    \end{equation}
    Indeed, from the definition of $j_x$ in~\eqref{eq: jx} and the bound $|q_x|,|p_x| \le \Const H_x^{1/2}$, we find $\sum_x |j_x| \le \Const \sum_x(H_{x-1}^2 + H_x^2) \le \Const H^2$. 
    Since $H^2$ is conserved and bounded by $(I_0 E_0)^2$, the bound \eqref{eq: bounded current} follows.  
    Therefore, we have arrived at
    \[
        M(t) |x(t) - \max I_0| \; \le \; \Const t
    \]
    which yields the statement upon enlarging the constant (recall $t\geq 1$.).  
\end{proof}

\subsection{A Priori Lower Bound on $M(t)$}
Given $0<\epsilon<\epsilon'<1$, let us define the (possibly infinite) random times
\begin{align}
    \randomt_\epsilon \; &= \; \min\{t\ge 0 : M(t) \le \epsilon\}, \label{eq: t epsilon}\\
    \randomt_{\epsilon,\epsilon'} \; &= \; \max\{t \le \randomt_\epsilon: M(t)\ge \epsilon'\}, \label{eq: t epsilon epsilon prime}
\end{align}
with the convention $\min\varnothing=+\infty$.
Note that these times are well-defined since $t\mapsto M(t)$ is continuous, being the maximum of Lipschitz-continuous functions with a uniform Lipschitz constant. 
We derive a first, deterministic, lower bound on $\randomt_\epsilon$: 

\begin{lemma}\label{lem: a priori bound}
    There exists a constant $\const$ such that, for any $\epsilon > 0$, 
    \begin{equation}\label{eq: deterministic lower bound t epsilon}
        \randomt_\epsilon \; \ge \; \frac{\const}{\epsilon}.
    \end{equation}
    In particular, $\randomt_\epsilon\to \infty$ deterministically as $\epsilon \to 0$.  
\end{lemma}

\begin{proof}
    Let $\epsilon > 0$. 
    Integrating the continuity equation~\eqref{eq: continuity equation} at the point $x(\randomt_{\epsilon,2\epsilon})$ yields 
    \[
        H_{x(\randomt_{\epsilon,2\epsilon})}(\randomt_{\epsilon}) -
        H_{x(\randomt_{\epsilon,2\epsilon})}(\randomt_{\epsilon,2\epsilon})
        \; = \; 
        \int_{\randomt_{\epsilon,2\epsilon}}^{\randomt_\epsilon}
        \dd s\, 
        \left(j_{x(\randomt_{\epsilon,2\epsilon})}(s) - j_{x(\randomt_{\epsilon,2\epsilon}) + 1}(s)\right).
    \]
    On the one hand, the absolute value of the left-hand side is lower bounded by $\epsilon$. 
    On the other hand, as argued in the proof of\eqref{eq: bounded current} above, $|j_x| \le \Const (H_{x-1}^2+H_x^2)$ for every $x\in\Z$, hence the absolute value of the right-hand side is upper bounded by $\Const \epsilon^2 (\randomt_\epsilon - \randomt_{\epsilon,2\epsilon})$.  
    This yields the claim. 
\end{proof}

\section{Approximate Fluctuation-Dissipation Decomposition}

To go beyond the crude bound \eqref{eq: deterministic lower bound t epsilon}, we will use an approximate fluctuation-dissipation decomposition of the current $j_x$ defined in \eqref{eq: jx}, as given by \eqref{eq: approximate commutator equation} below.
We do not expect $j_x$ to be written as a total derivative of a local function, since this would imply the total absence of macroscopic transport. 
In other words, we do not anticipate that the equation $j_x = - \{H,u_x\}$ can be solved with $u_x$ depending on a finite number of coordinates around the site $x$
(non-local solutions are easily found, e.g., $u_x = \sum_{y \leq x-1} H_y$, reflecting energy conservation). Instead, we will construct a perturbative expansion and find two local functions $u_x$ and $g_x$ satisfying 
\begin{equation}\label{eq: approximate commutator equation}
    j_x \; = \; -\{H,u_x\} + g_x. 
\end{equation}
Importantly, $g_x$ can be constructed as a homogeneous polynomial in $q$ and $p$ of arbitrarily high degree, at the price of enlarging the support of $u_x$ and $g_x$. 
At low values of the local energy density $\varrho$, this term will thus vanish as $\varrho^{d/2}$ with $d$ the degree of $g_x$. 



We follow the strategy developed in \cite{huveneers_drastic_2013,de_roeck_long_persistence_2023}.
The expansion is set up as follows.
Let $f^{(1)}_x = j_x$ and let us assume that we can define recursively observables $u^{(i)}_x$ and $f^{(i+1)}_x$ for $i\ge 1$ satisfying
\begin{equation}\label{eq: perturbative scheme}
    f^{(i)}_x \; = \; -\{H_{\har},u^{(i)}_x\}, 
    \qquad 
    f^{(i+1)}_x \; = \; \{H_{\an},u^{(i)}_x\}
\end{equation}
with 
\begin{equation}\label{eq: H har and an}
    H_\har \; = \; \sum_{x\in\Z} \frac{p_x^2}{2} + \frac{\omega_x^2 q_x^2}{2}, 
    \qquad 
    H_\an \; = \; \sum_{x\in\Z} U(q_{x-1} - q_{x}).
\end{equation}
Given an integer $n\ge 1$, we define 
\begin{equation}\label{eq: solutions u and f order n}
    u_x \; = \; u^{(1)}_x + \dots + u^{(n)}_x, \qquad g_x \; = \; f^{(n+1)}_x
\end{equation}
and we verify that $u_x$ and $f_x$ solve \eqref{eq: approximate commutator equation}. 

\subsection{Diagonalization of the Harmonic Dynamics}\label{sec: change of variables}
Let us introduce a change of coordinates: given $x\in\Z$, let
\begin{equation}\label{eq: change of variables}
    \begin{cases}
    a_x^\pm \; = \; \sqrt{\frac{\omega_x}{2}} \left( q_x \mp \imag\frac{p_x}{\omega_x} \right),\\
    q_x \; = \; \frac{1}{\sqrt{2\omega_x}}(a_x^+ + a_x^-), 
    \qquad 
    p_x \; = \; \imag \sqrt{\frac{\omega_x}{2}}(a_x^+ - a_x^-).
    \end{cases}
\end{equation}
Given $\bx = (x_1,\dots,x_d)\in\Z^d$ and $\bsigma=(\sigma_1,\dots,\sigma_d)\in\{\pm 1\}^d$
for some integer $d\ge 1$, let 
\begin{equation}\label{eq: basic monomial}
    a_\bx^\bsigma \; = \; a_{x_1}^{\sigma_1}\dots a_{x_d}^{\sigma_d}.
\end{equation}
We write also $d = d(\bx,\bsigma)$.
These coordinates have the advantage that the map $u\mapsto \{H_\har,u\}$ becomes diagonal. 
Indeed, the rules
\begin{equation}\label{eq: poisson bracket rules}
    \{a^\sigma_x,a^{\sigma'}_{x'}\} 
    \; = \; 
    -\imag \sigma \delta(\sigma+\sigma') \delta(x-x'), 
    \qquad 
    \{\varphi_1,\varphi_2\varphi_3\} 
    \; = \; 
    \{\varphi_1,\varphi_2\} \varphi_3 + \varphi_2 \{\varphi_1,\varphi_3\}
\end{equation}
hold for any $x,x'\in\Z$ and any smooth observables $\varphi_1,\varphi_2,\varphi_3$.
Here and in the remainder of the paper, to avoid the appearance of excessive subscripts, we use $\delta(\cdot)$ to denote the Kronecker delta function, i.e.\@ $\delta(x) = 1$ if $x=0$ and $\delta(x)=0$ otherwise.
Since 
\begin{equation}\label{eq: H har}
H_\har \; = \; \sum_{x\in\Z} \omega_x a_x^+ a_x^-,
\end{equation}
we compute that given some $(\bx,\bsigma)$, 
\begin{equation}\label{eq: def denominators}
    \{H_\har,a^\bsigma_\bx\} 
    \; = \; \imag \Delta(\bx,\bsigma) a^\bsigma_\bx
    \qquad \text{with} \qquad 
    \Delta(\bx,\bsigma) \; = \; \sigma_{1} \omega_{x_1} + \dots +\sigma_{d} \omega_{x_d}
\end{equation}
and $d=d(\bx,\bsigma)$.

We can now identify all monomials of the form \eqref{eq: basic monomial} that do not lie in the image of the map $u \mapsto -\{H_\har, u\}$.
We say that $(\bx,\bsigma)\in\mathcal S$ if and only if $d(\bx,\bsigma)$ is even and if it is possible to partition the set $(1,\dots,d)$ into pairs $(l_1,m_1), \dots, (l_{d/2},m_{d/2})$ such that $x_{l_k} = x_{m_k}$ and $\sigma_{l_k} = - \sigma_{m_k}$ for $1 \le k \le d/2$. 
One observes that $\Delta(\bx,\bsigma) = 0$ a.s.\@ if and only if $(\bx,\bsigma)\in\mathcal S$, hence $a_\bx^\bsigma$ belongs to the image of $u\mapsto -\{H_\har,u\}$ a.s.\@ if and only if $(\bx,\bsigma)\notin\mathcal S$.

\subsection{Good Behavior of the Perturbative Expansion}
Let us unveil a symmetry that ensures the functions $f_x^{(i)}$ are in the image of the map $u \mapsto -\{H_\har, u\}$ almost surely, for all $x \in \Z$ and $i \geq 1$.
Let us start with introducing some notions. 
Let a homogeneous polynomial of degree $d\ge 1$ be a function of the form
\[
    f \; = \; \sum_{(\bx,\bsigma)\in\Z^d\times\{\pm 1 \}^d} F(\bx,\bsigma) a_\bx^\bsigma.
\]
Note that $0$ can therefore be considered as a polynomial of arbitrary degree.
Given a homogeneous polynomial of degree $d$, we define 
\[
    \mathcal P f \; = \; 
    \sum_{(\bx,\bsigma)\in\Z^d\times\{\pm 1 \}^d} F(\bx,\bsigma) a_\bx^{-\bsigma}
    \qquad \text{with} \qquad -\bsigma \; = \; (-\sigma_1,\dots,-\sigma_d).
\]
A polynomial $f$ is said to be $\mathcal P$-symmetric if $\mathcal P f = f$ and $\mathcal P$-skew-symmetric if $\mathcal P f = - f$. 

We now state the announced symmetry: 
\begin{lemma}
If $f$ is a $\mathcal P$-symmetric homogeneous polynomial of degree $d$ and $f'$ a $\mathcal P$-symmetric homogeneous polynomial of degree $d'$, then $\{f,f'\}$ is a $\mathcal P$-skew-symmetric homogeneous polynomial of degree $d+d'-2$. 
Similarly, if $f$ is a $\mathcal P$-symmetric homogeneous polynomial of degree $d$ and $f'$ a $\mathcal P$-skew-symmetric homogeneous polynomial of degree $d'$, then $\{f,f'\}$ is a $\mathcal P$-symmetric homogeneous polynomial of degree $d+d'-2$.
\end{lemma}

\begin{proof}
This follows from the rules \eqref{eq: poisson bracket rules}, which imply in particular that $\{a_\bx^\bsigma,a_\by^\btau\} \pm \{a_\bx^{-\bsigma},a_\by^{-\btau}\}$ is $\mathcal{P}$-skew-symmetric if $\pm = +$, and $\mathcal{P}$-symmetric if $\pm = -$, for any $(\bx,\bsigma)$ and $(\by,\btau)$.
\end{proof}

The key consequence for us is that the iteration \eqref{eq: perturbative scheme} is well-defined.
Specifically, we ensure that $f_x^{(i)}$ can be taken to be a $\mathcal P$-skew-symmetric homogeneous polynomial of degree $4 + 2(i-1)$ for every $x \in \Z$ and $i \ge 1$, 
implying that $f_x^{(i)}$ lies in the image of the map $u \mapsto -\{H_\har, u\}$, since any monomial $a_\bx^\bsigma$ is $\mathcal P$-symmetric when $(\bx, \bsigma) \in \mathcal S$.
(We say ``can be taken to be'' rather than ``is'' because it is always possible to modify $u_x^{(i)}$ by adding elements from the kernel of $u\mapsto -\{H_{\har},u\}$, which would in turn modify $f_x^{(i+1)}$.) 

To check this, we proceed to some explicit computations: 
\begin{enumerate}
\item 
By~\eqref{eq: jx} and \eqref{eq: change of variables}, 
$f^{(1)}$ is $\mathcal P$-skew-symmetric, homogeneous, and of degree 4. 
More precisely, it takes the form 
\begin{equation}\label{eq: current new variables}
    f^{(1)}_x \; = \; j_x \; = \; \sum_{(\bx,\bsigma)} J_x(\bx,\bsigma) a_\bx^\bsigma
\end{equation}
where the sum runs over all $(\bx,\bsigma)$ such that $d(\bx,\bsigma) = 4$.
Note that the coefficients $J_x(\bx,\bsigma)$ are not uniquely determined by \eqref{eq: jx} and \eqref{eq: change of variables}.
We take them such that $J_x(\bx',\bsigma')=J_x(\bx,\bsigma)$ whenever $a_{\bx'}^{\bsigma'} = a_{\bx}^{\bsigma}$.
In paticular, they vanish unless $\bx = (x_1,\dots,x_4)$ is such that $x_1,\dots,x_4 \in\{x-1,x\}$. 

\item 
It follows directly from \eqref{eq: H har} that $H_\har$ is a homogeneous $\mathcal P$-symmetric polynomial of degree 2. 

\item 
Using \eqref{eq: H har and an} and \eqref{eq: change of variables},  we conclude that $H_\an$ is a homogeneous $\mathcal P$-symmetric polynomial of degree $4$. More precisely, 
\begin{equation}\label{eq: H an expansion}
    H_\an \; = \; \sum_{\bx,\bsigma}V(\bx,\bsigma) a_\bx^\bsigma
\end{equation}
where the sum runs again over all $(\bx,\bsigma)$ such that $d(\bx,\bsigma) = 4$. 
Here as well, we take the coefficients $V(\bx,\bsigma)$ such that $V(\bx',\bsigma')=V(\bx,\bsigma)$ whenever $a_{\bx'}^{\bsigma'} = a_\bx^\bsigma$. 
They vanish unless $\bx = (x_1,\dots,x_4)$ is such that $x_1,\dots,x_4 \in\{y-1,y\}$ for some $y\in\Z$.
\end{enumerate}
The claim that that $f_x^{(i)}$ can be taken to be a $\mathcal P$-skew-symmetric polynomial of degree $4 + 2(i-1) = 2(i+1)$ then readily follows from the iteration~\eqref{eq: perturbative scheme}, from the action of $u\to\{H_\har,u\}$ in \eqref{eq: def denominators} and the symmetry discussed above.



\subsection{Explicit Representations}

We now provide explicit expressions for the functions $u_x^{(i)}$ and $f_x^{(i)}$ that solve the iteration \eqref{eq: perturbative scheme}. Based on the previous considerations, this process is, in principle, straightforward: the relation $f_x^{(i)} = - \{H_\har, u_x^{(i)}\}$ can be inverted using \eqref{eq: def denominators} to express $u_x^{(i)}$ in terms of $f_x^{(i)}$, and $f_x^{(i+1)}$ can then be explicitly determined from $f_x^{(i+1)} = \{H_\an, u_x^{(i)}\}$ using the rules in \eqref{eq: poisson bracket rules}. However, to write a sufficiently explicit expression that will be useful later, we first need to introduce some additional terminology and notation:

\begin{enumerate}
\item 
Given $i\ge 1$, let $(\bx,\bsigma)\in\Z^{4i}\times\{\pm 1\}^{4i}$ and write
\[
\begin{cases}
\bx \; = \; (\bx_1,\dots,\bx_i), \quad x_l = (x_{l,1},\dots,x_{l,4}), \quad 1\le l \le i,\\
\bsigma \; = \; (\bsigma_1,\dots,\bsigma_i), 
\quad \sigma_l = (\sigma_{l,1},\dots,\sigma_{l,4}), \quad 1\le l \le i.
\end{cases}
\]
Given $x\in\Z$, we define the set $\mathcal X_i(x)$ of such pairs $(\bx,\bsigma)$ with the following restrictions: 
\begin{enumerate}
    \item 
    $x_{1,1},\dots,x_{1,4}\in\{x-1,x\}$ and, for every $l\in\{1,\dots,i\}$, there exists $y(l)\in\Z$ such that $x_{l1},\dots,x_{l4}\in\{y(l)-1,y(l)\}$, 

    \item 
    $
        (\bx_1,\bsigma_1)\notin \mathcal S, \;
        ((\bx_1,\bx_2),(\bsigma_1,\bsigma_2))\notin \mathcal S, \; \dots, \;
         ((\bx_1,\dots,\bx_i),(\bsigma_1,\dots,\bsigma_i))\notin \mathcal S.
    $
\end{enumerate}
Moreover, for $(\bx,\bsigma)\in\mathcal X_i(x)$, we define the \emph{support} of $(\bx,\bsigma)$ as the smallest interval $I\subset\Z$ (for the inclusion) such that $x_{1,1},\dots, x_{i,4}\in I$.\\
The elements $(\bx,\bsigma)\in\mathcal{X}_i(x)$ will be used to index all monomials appearing in the expansion of $f_x^{(i)}$ and $u^{(i)}_x$. This can be understood given that the current $j_x$ is a homogeneous polynomial of degree 4, cf.\eqref{eq: current new variables}; that $u^{(i)}_x$ involves the same monomials as $f^{(i)}_x$ by \eqref{eq: def denominators}; and that each new step in the perturbation introduces a new $H_\an$, cf.\eqref{eq: perturbative scheme}, which is also a homogeneous polynomial of degree 4 by \eqref{eq: H an expansion}.

\item
Given $i\ge 2$, let us define the set of \emph{contractions} $\mathcal C_i$ as the set of all $4(i-1)$-tuples of positive integers of the form
\[
    s \; = \; \Big(\big(2,l(2);k(2),l'(2)\big),\dots,\big(i,l(i);k(i),l'(i)\big)\Big)
\]
with $1 \le l(j),l'(j)\le4$ and $1 \le k(j)<j$ for $2\le j \le i$, and with the constraint that the couples $(2,l(2)),(k(2),l'(2)), \dots, (i,l(i)),(k(i),l'(i))$ are all different.\\
The word contraction refers here to the use of the first rule in \eqref{eq: poisson bracket rules}, where we say that the indices $(x,\sigma)$ and $(x',\sigma')$ are \emph{contracted} through the Poisson bracket. Since, in reaching order $i$ in perturbation, we take the Poisson bracket with $H_\an$ a total of $i-1$ times, cf.~\eqref{eq: perturbative scheme}, we need to specify which indices are contracted at each step. This is the purpose of specifying $s$, as defined above and further described below.

\item 
Given $i\ge 2$, we say that $(\bx,\bsigma)\in\mathcal X_i(x)$ can be \emph{contracted} through a contraction $s\in\mathcal C_i$ if
\begin{equation}\label{eq: contraction satisfied}
    x_{2,l(2)} = x_{k(2),l'(2)}, \dots, x_{i,l(i)} = x_{k(i),l'(i)}, 
    \quad
    \sigma_{2,l(2)} = -\sigma_{k(2),l'(2)}, \dots, \sigma_{i,l(i)} = -\sigma_{k(i),l'(i)}
\end{equation}
with $s = ((2,l(2);k(2),l'(2)),\dots,(i,l(i);k(i),l'(i)))$. 
We let also $\Phi$ be a boolean function, i.e.\@ taking values in $\{0,1\}$, such that 
\[
    \Phi (\bx,\bsigma;s) \; = \; 1 
\]
if and only if $(\bx,\bsigma)$ is contracted by $s$.
Further, if $(\bx,\bsigma)$ can be contracted by a contraction $s\in\mathcal C_i$, 
we denote by $\tilde a_\bx^\bsigma(s)$ the monomial defined as in \eqref{eq: basic monomial} where the product is restricted to the non-contracted indices, i.e.\@ we omit the factors $a_{x_{j,l(j)}}^{\sigma_{j,l(j)}} a_{x_{j,l(j)}}^{-\sigma_{j,l(j)}}$ for $2 \le j \le i$.
We observe that while the degree of $a_\bx^\bsigma$ is $4i$, the degree of $\tilde a_\bx^\bsigma$ is $4 + 2(i-1)=2(i+1)$. 
\end{enumerate}

We first note
\begin{lemma}\label{lem: combinatorics}
    Let $x\in\Z$ and $i\ge 1$. 
    \begin{enumerate}
    \item
    There exists a constant $\Const$ such that 
    \begin{equation*}
    \left|\left\{ (\bx,\bsigma)\in\mathcal X_i(x) : (\bx,\bsigma)\text{ can be contracted} \right\}\right| 
    \;\le\; 
    \sum_{(\bx,\bsigma)\in\mathcal X_i(x)}
    \sum_{s\in\mathcal C_i}
    \Phi (\bx,\bsigma;s)
    \; \le \; 
    \ed^{\Const i \ln (i+1)}. 
    \end{equation*}

    \item
    If $(\bx,\bsigma)\in\mathcal X_i(x)$ can be contracted, then $(\bx,\bsigma)$ is supported on the interval $[x-i,x+i-1]\cap\Z$.
    \end{enumerate}
\end{lemma}
\begin{proof}
    The cardinality of $\mathcal C_i$ is bounded by $\Const^i i! $ and, given a contraction $s\in\mathcal C_i$, the number of elements $(\bx,\bsigma)\in\mathcal X_i(x)$ such that $\Phi(\bx,\bsigma;s)=1$ is bounded by $\Const^i$. 
    This yields the first point.
    The second point is a consequence of the locality imposed by \eqref{eq: contraction satisfied}.
\end{proof}

We state now the main result of this section:

\begin{proposition} \label{prop: main expression f u}
    For $i\ge 1$, let 
    \begin{multline}\label{eq: explicit f i}
    f^{(i)}_x \; = \; \sum_{(\bx,\bsigma)\in\mathcal X_i(x)}
    \sum_{s\in\mathcal C_i}
    \Phi (\bx,\bsigma;s)\\
    \frac{ 
    \sigma_{2,k(2)}\dots \sigma_{i,k(i)}  J_x(\bx_1,\bsigma_1)V(\bx_2,\bsigma_2)\dots V(\bx_i,\bsigma_i)}{\Delta(\bx_1,\bsigma_1)\Delta((\bx_1,\bx_2),(\bsigma_1,\bsigma_2))\dots 
    \Delta((\bx_1,\dots,\bx_{i-1}),(\bsigma_1,\dots,\bsigma_{i-1}))} \tilde a_\bx^\bsigma(s)
    \end{multline}
    and
    \begin{multline}\label{eq: explicit u i}
    u^{(i)}_x \; = \; \sum_{(\bx,\bsigma)\in\mathcal X_i(x)}
    \sum_{s\in\mathcal C_i}
    \Phi(\bx,\bsigma;s)\\
    \frac{\imag 
    \sigma_{2,k(2)}\dots \sigma_{i,k(i)}  J_x(\bx_1,\bsigma_1)V(\bx_2,\bsigma_2)\dots V(\bx_i,\bsigma_i)}{\Delta(\bx_1,\bsigma_1)\Delta((\bx_1,\bx_2),(\bsigma_1,\bsigma_2))\dots 
    \Delta((\bx_1,\dots,\bx_{i}),(\bsigma_1,\dots,\bsigma_{i}))} \tilde a_\bx^\bsigma(s),
    \end{multline}
    {where $J_x$, $V$, and $\Delta$ are defined respectively in \eqref{eq: current new variables}, \eqref{eq: H an expansion} and \eqref{eq: def denominators}.}
    These functions solve the iteration~\eqref{eq: perturbative scheme}. 
\end{proposition}

For $i=1$, the above expressions are interpreted as follows: the sum over contractions, the factors $\sigma$ the factors of $V$, and denominators in \eqref{eq: explicit f i} and function $\Phi$ are all absent, and $\tilde a_\bx^\bsigma(s)$ stands for $a_\bx^\bsigma$. 

\begin{proof}
    First, the expression \eqref{eq: explicit f i} for $f^{(1)}_x$ is a rewriting of \eqref{eq: current new variables}, where we note that the summation can be restricted to $(\bx,\bsigma)\notin\mathcal S$  since $j_x$ is $\mathcal P$-skew-symmetric.
    Second, the expression \eqref{eq: explicit u i} for $u^{(i)}_x$ follows from the expression \eqref{eq: explicit f i} for $f^{(i)}_x$ for all $i\ge 1$. 
    Indeed, by the spectral decomposition \eqref{eq: def denominators} and by the definition of $\tilde a_\bx^\bsigma (s)$, we find 
    \[
        -\{ H_\har,\tilde a_\bx^\bsigma (s)\} \; = \; - \imag \Delta (\bx,\bsigma) \tilde a_\bx^\bsigma (s)
    \]
    for any $(\bx,\bsigma)\in\mathcal X_i(x)$ and $s\in\mathcal C_i$.
    In other words, $\tilde a_\bx^\bsigma(s)$ and $a_\bx^\bsigma$ are eigenvectors of the map $u\mapsto -\{H_{\har},u\}$ with the same eigenvalue $\imag \Delta(\bx,\bsigma)$.
    Therefore, the expression \eqref{eq: explicit u i} for $u^{(i)}_x$ ensures that the required relation $f^{(i)}_x = - \{ H_{\har}, u_x^{(i)} \}$ in \eqref{eq: perturbative scheme} is satisfied.

We still need to show that the expression~\eqref{eq: explicit u i} for \( u_x^{(i)} \) leads to the expression~\eqref{eq: explicit f i} for \( f_x^{(i+1)} \). 
We start with the second expression in \eqref{eq: perturbative scheme}, substituting the representation \eqref{eq: H an expansion} for \( H_\an \) and the representation~\eqref{eq: explicit u i} for \( u_x^{(i)} \). To simplify notation, we denote the entire fraction in the representation~\eqref{eq: explicit u i} of \( u_x^{(i)} \) as \( U_x^{(i)}(\bx, \bsigma) \). This gives
\begin{multline*}
f^{(i+1)}_x = \{H_\an, u_x^{(i)}\} = 
\sum_{(\bx',\bsigma')} \sum_{(\bx'',\bsigma'') \in \mathcal{X}_i} 
\sum_{l, (j, l(j))} \sum_{s'' \in \mathcal{C}_i} \\
\delta(x'_l - x''_{j,l(j)}) \delta(\sigma'_l + \sigma''_{j,l(j)}) 
\Phi(\bx'', \bsigma''; s'') (-\imag) \sigma'_l V(\bx', \bsigma') 
U^{(i)}(\bx'', \bsigma'') 
\tilde{a}_{(\bx'', \bx')}^{(\bsigma'', \bsigma')}(s'', l).
\end{multline*}
Here, primed variables are used for the components of \( H_\an \), and double-primed variables are used for the components of \( u_x^{(i)} \). The third summation is over \( 1 \leq l \leq 4 \) and all pairs \( (j, l(j)) \) indexing the non-contracted variables of \( (\bx'', \bsigma'') \). Finally, \( \tilde{a}_{(\bx'', \bx')}^{(\bsigma'', \bsigma')}(s'', l) \) represents the product \( a_{\bx'}^{\bsigma'} \tilde{a}_{\bx''}^{\bsigma''} / a_{x'_l}^+ a_{x'_l}^- \).

    
    In the above expression, the summations over $(\bx',\bsigma')$ and $(\bx'',\bsigma'')\in\mathcal X_i(x)$ can be gathered into a single summation over $(\bx,\bsigma)\in\mathcal X_{i+1}$ with $\bx = (\bx'',\bx')$ and $\bsigma = (\bsigma'',\bsigma')$
    provided we restrict $(\bx,\bsigma)\notin\mathcal S$, which we can do knowing that $f^{(i+1)}$ is $\mathcal P$-skew-symmetric. 
    Similarly, the two summations over contractions can be gathered into a single summation over contractions $s\in\mathcal C_{i+1}$, while the the product over the two delta factors and the boolean factor $\Phi(\bx'',\bsigma'',s'')$ writes as a single factor $\Phi (\bx,\bsigma;s)$.  
    Finally, the product $ (-\imag)\sigma'_l V(\bx',\bsigma') U^{(i)}(\bx'',\bsigma'')$ is identified with the whole fraction in \eqref{eq: explicit f i} with $i$ replaced by $i+1$ and the monomial $ \tilde a_{(\bx'',\bx')}^{(\bsigma'',\bsigma')}(s'',l)$ is identified with $\tilde a_\bx^\bsigma$. 
    This concludes the proof of the proposition. 
    %
\end{proof}

\section{Probabilistic Bounds}
We deal here with the probability of \emph{resonances}, corresponding to events where the denominator in the expressions \eqref{eq: explicit f i} or \eqref{eq: explicit u i} becomes smaller than some threshold. 
Let a positive integer $n$ be given for this whole section, that corresponds to the order in perturbation appearing in \eqref{eq: solutions u and f order n}. 
Let also $0<\delta<1$ be our resonant threshold.

\subsection{Probability of Resonances}
We state the primary probabilistic result to be applied in this section:

\begin{lemma}\label{lem: probability small denominator}
    Let $d\ge 1$ be an integer and let $(\bx,\bsigma)\notin \mathcal S$ with $d(\bx,\bsigma) = d$. There exists a constant $\Const$ such that
    \[
        \Proba(|\Delta(\bx,\bsigma)|\le\delta) \; \le \; \Const  \delta. 
    \]
\end{lemma}
\begin{proof}
    Since $(\bx,\bsigma)\notin\mathcal S$, $\Delta(\bx,\bsigma)$ takes the form
    \[
    \Delta(\bx,\bsigma) 
    \; = \;
    r_1 \omega_{y_1} + \dots + r_m \omega_{y_m}
    \]
    for some positive integer $m\le d$, some non-zero integers $r_1,\dots,r_m$ with $\sum_{i=1}^m|r_i|\le d$, and different points $y_1,\dots,y_m$.
    Using the change of variables
    \[
    \nu_1 = \omega_{y_1},\quad \dots,\quad \nu_{m-1} = \omega_{y_{m-1}}, 
    \quad
    \nu_m = r_1 \omega_{y_1} + \dots + r_m \omega_{y_m},
    \]
    and denoting by $\rho$ the density of $\omega_x$, 
    we compute
    \begin{align*}
    \Proba(|\Delta(\bx,\bsigma)|\le \delta)
    \; &= \; 
    \int\dd \omega_{y_1} \rho(\omega_{y_1}) \dots \int \dd \omega_{y_m}\rho(\omega_{y_m})
    1_{\{|\Delta(\bx,\bsigma)|\le \delta\}}\\
    \; &\le \; 
    \int \dd \nu_1 \rho(\nu_1) \dots \int \dd \nu_{m-1} \rho(\nu_{m-1})
    \int\dd \nu_m\rho(\omega_{y_m})
    1_{\{|\nu_m|\le \delta\}} \\
    \; &\le \; 
    \Const \int \dd \nu_1 \rho(\nu_1) \dots \int \dd \nu_{m-1} \rho(\nu_{m-1})
    \int\dd \nu_m 1_{\{|\nu_m|\le \delta\}} \\
    \; & \le \; \Const \delta
    \end{align*}
    where we have used the fact that $|\det (\partial \omega/\partial\nu)| = 1/|r_m| \le 1$, where $\partial\omega/\partial\nu$ denotes the Jacobian matrix.
    %
    %
    %
    %
\end{proof}

\subsection{Resonant Points and Resonant Intervals}
We say that a point $x\in\Z$ is \emph{$(n,\delta)$-resonant} if
\[
    |\Delta(\bx,\bsigma)| \;\le\; \delta. 
\]
for some $(\bx,\bsigma)\in\mathcal X_i(x)$, with $i\in\{1,\dots,n\}$,  that can be contracted, i.e.\@ for which \eqref{eq: contraction satisfied} is satisfied for some contraction in $\mathcal C_i$.
Let us make three observations related to this definition: 
\begin{enumerate}
    \item 
    If a point $x$ is not $(n,\delta)$-resonant, then the denominators in \eqref{eq: explicit f i} and \eqref{eq: explicit u i} are lower-bounded in absolute value by $\delta^{i-1}$ and $\delta^i$ respectively. 

    \item\label{item: probability x non resonant}
    There exists a constant $\Const_1$ such that, given any $x\in\Z$,
    \begin{equation}\label{eq: probability x resonant}
    \Proba(\text{$x$ is $(n,\delta)$-resonant}) 
    \; \le \; 
    \ed^{\Const_1 n \ln (n+1)} \delta.
    \end{equation}
    This follows from the first point of Lemma~\ref{lem: combinatorics} and Lemma~\ref{lem: probability small denominator}.

    \item\label{item: event x resonant}
    The event that $x$ is $(n,\delta)$-resonant only depends only on the disorder in the interval $[x-n,x+n-1]\subset\Z$. 
    This follows from the second point of Lemma~\ref{lem: combinatorics}. 
\end{enumerate}

Next, we say that an interval $I\subset \Z$ is $(n,\delta)$-resonant if every point in $I$ is $(n,\delta)$-resonant. 
Given $x\in\Z$, we denote by $R(n,\delta;x)\subset\Z$ the largest interval (for the inclusion) such that $x\in R(n,\delta;x)$ and $R(n,\delta;x)$ is $(n,\delta)$-resonant. 
Given $x\in\Z$ and $\ell \ge 1$, we find that 
\begin{align}\label{eq: probability interval resonant}
    \Proba(|R(n,\delta;x)|\ge \ell) 
    \; &\leq \; \sum_{y=x-\ell}^x \Proba\left([y,y+\ell]\subset R(n,\delta;x)\right) \nonumber \\
    &\leq \; \sum_{y=x-\ell}^x \Proba\left(\bigcap_{\substack{z\in[y,y+\ell] \\ z-y \text{ is a multiple of $2(n+1)$}}} \{z \text{ is $(n,\delta)$-resonant}\}\right)\nonumber\\
    &\le \; 
    \ell \left( \ed^{\Const_1 n \ln (n+1)} \delta \right)^{\ell/2(n+1)} ,
\end{align}
where the last inequality follows from points~\ref{item: probability x non resonant} and \ref{item: event x resonant} above. 

We now show the main result of this section. 
Recall the definition of $x(t)$ in \eqref{eq: def of x(t)}, as well as the definitions of $\randomt_\epsilon$ and $\randomt_{\epsilon,\epsilon'}$ in \eqref{eq: t epsilon} and \eqref{eq: t epsilon epsilon prime} respectively. Let also $c$ be the constant in Lemma~\ref{lem: a priori bound}.

\begin{proposition}\label{pro: conclusion probabilistic bound}
    Provided $n$ and $\delta$ are such that 
    \begin{equation}\label{eq: condition n and delta}
        \left( \ed^{(4+\Const_1\ln (n+1))(n+1)} \delta \right)^{\frac{1}{2(n+1)}}
        \;\le\; \frac12,  
    \end{equation}
    and provided $\epsilon'$ is small enough so that $\ln(c/\epsilon')\ge 1$, then
    \[
    \Proba\left(|R(n,\delta;x(\randomt_{\epsilon,\epsilon'}))| \ge \ln \randomt_{\epsilon,\epsilon'}\right)
    \; \le \; 
    \frac{\Const}{\epsilon'} 
    \left( \ed^{(4+\Const_1\ln (n+1))(n+1)} \delta \right)^{\frac{\ln(c/\epsilon')}{2(n+1)}}.
    \]
\end{proposition}
\begin{proof}
    For brevity, we write $R(x)$ instead of $R(n,\delta;x)$. 
    By Lemma~\ref{lem: a priori bound}, we know that $\randomt_{\epsilon,\epsilon'} \ge \randomt_{\epsilon'} \ge {c}/{\epsilon'}$. 
    Hence we can bound 
    \begin{equation}\label{eq: 1st bound proof probabilistic}
    \Proba\left(|R(x(\randomt_{\epsilon,\epsilon'}))| \ge \ln \randomt_{\epsilon,\epsilon'}\right)
    \; \le \; 
    \sum_{j\ge \ln (c/\epsilon')} \Proba\left(R(x(\randomt_{\epsilon,\epsilon'})) \ge j,\; j \le \ln \randomt_{\epsilon,\epsilon'} < j+1 \right)
    \end{equation}
    where the sum runs over all $j$ of the form $j=\ln (c/\epsilon') + u$ with $u$ a non-negative integer. 
    By Lemma~\ref{lem: light cone} and since $M(\randomt_{\epsilon,\epsilon'}) = \epsilon'$, the constraint $\randomt_{\epsilon,\epsilon'}< \ed^{j+1}$ implies $|x(\randomt_{\epsilon,\epsilon'})|\le \Const \ed^{j}/\epsilon'$. 
    Using~\eqref{eq: probability interval resonant} by taking a union bound over all possible positions for $x(\randomt_{\epsilon,\epsilon'})$ yields 
    \begin{align*}
    \Proba\left( R(x(\randomt_{\epsilon,\epsilon'})) \ge j,\; \randomt_{\epsilon,\epsilon'} < \ed^{j+1} \right) 
    \; &\le \; \sum_{x:|x|\le \Const \ed^j/\epsilon'} 
    \Proba (R(x) \ge j) 
    \; \le \; 
    \frac{\Const \ed^{j} j}{\epsilon'} 
     \left( \ed^{\Const_1 n \ln (n+1)} \delta \right)^{\frac{j}{2(n+1)}} \\
    \; &\le \; 
    \frac{\Const}{\epsilon'} 
     \left( \ed^{(4+\Const_1\ln (n+1))(n+1)} \delta \right)^{\frac{j}{2(n+1)}}  .
    \end{align*}
    Inserting this bound into \eqref{eq: 1st bound proof probabilistic}, and using the assumption~\eqref{eq: condition n and delta} as well as $\ln (c/\epsilon')\ge 1$, yields the claim. 
\end{proof}

\section{Concluding the Proof}\label{sec: concluding the proof}

From now on, we relate the parameters $n$, $\delta$ and $\epsilon'$ to $\epsilon$: 
\begin{equation}\label{eq: defining parameter}
    n \; = \; \left\lceil(\ln (1/\epsilon))^{\theta}\right\rceil, 
    \qquad 
    \delta \; = \; \epsilon^{\eta}, 
    \qquad 
    \epsilon' \; = \; \epsilon^{1-\eta}.
\end{equation}
Here, we take 
\begin{equation}\label{eq: values of eta and theta}
    \theta \; = \; \frac13 \qquad \text{and} \qquad \eta \; = \; \frac1{10}
\end{equation}
but we find it convenient to keep using the symbols $\theta$ and $\eta$ in most places. 

\subsection{Improved Lower Bound on $M(t)$}
We establish here the following lemma, which represents a significant improvement over Lemma~\ref{lem: a priori bound} on an event of nonzero probability:
\begin{lemma}\label{lem: probabilistic improvement on the a priori bound}
  
    There exists a constant $\epsilon_0$ such that for $\epsilon \le \epsilon_0$,  
    \begin{equation}\label{eq: definition of phi}
    \randomt_\epsilon \; \ge \; \varphi(\epsilon) 
    \; := \; \ed^{\frac{1}{2} (\ln(1/\epsilon))^{1+\theta}}
    \end{equation}
    on the event $|R(n,\delta;x(\randomt_{\epsilon,\epsilon'}))| \le \ln \randomt_{\epsilon,\epsilon'}$.
\end{lemma}
\begin{proof}
    Given an interval $I\subset \Z$, we write $H_I = \sum_{x\in I}H_x$. 
    For convenience, let us denote the (random) interval $R(n,\delta,x(\randomt_{\epsilon,\epsilon'}))$ as $[a,b]$. 
    The continuity equation~\eqref{eq: continuity equation} yields
    \[
        H_{[a-1,b]}(\randomt_\epsilon) - H_{[a-1,b]}(\randomt_{\epsilon,\epsilon'})
        \; = \; 
        \int_{\randomt_{\epsilon,\epsilon'}}^{\randomt_\epsilon} \dd s
        \left( j_{a-1}(s) - j_{b+1}(s) \right)
    \]
    and, using the decomposition~\eqref{eq: approximate commutator equation}, this implies
    \begin{multline}\label{eq: main argument}
    \left(H_{[a-1,b]}(\randomt_\epsilon) - H_{[a-1,b]}(\randomt_{\epsilon,\epsilon'})\right)
    + \left( u_{a-1}(\randomt_\epsilon) - u_{a-1}(\randomt_{\epsilon,\epsilon'}) \right)
    - \left( u_{b+1}(\randomt_\epsilon) - u_{b+1}(\randomt_{\epsilon,\epsilon'})\right)\\
    \; = \; 
    \int_{\randomt_{\epsilon,\epsilon'}}^{\randomt_\epsilon} \dd s
        \left( g_{a-1}(s) - g_{b+1}(s) \right).
    \end{multline}

    Let us first derive an upper bound on the absolute value of the right-hand side of \eqref{eq: main argument}. 
    An explicit expression for $g_{a-1}$ and $g_{b+1}$ is provided in \eqref{eq: explicit f i} with $i=n+1$. 
    First, since $a-1$ and $b+1$ are not $(n,\delta)$-resonant, the denominators in this expression can be lower bounded in absolute value as $\delta^{n+1}$. 
    Next, we upper-bound $|\tilde a_\bx^\bsigma(s)|\le \Const^n (\epsilon')^{n+2}$, which follows from the fact that the degree of $\tilde a_\bx^\bsigma(s)$ is $2(n+2)$, and from the generic bound 
    \[
        |a_x^\sigma| \; \le \; 
        \Const H_x^{1/2} \; \le \; \Const M(t)^{1/2} \;\le\; \Const (\epsilon')^{1/2},
    \]
    valid for any $t\in[\randomt_{\epsilon,\epsilon'},\randomt_\epsilon]$, 
    any $x\in\Z$ and any $\sigma=\pm 1$.
    Finally, using the bound in the first point of Lemma~\ref{lem: combinatorics} to bound the number of terms in \eqref{eq: explicit f i}, we arrive at 
    \begin{align}
    |\text{r.h.s.\@ of \eqref{eq: main argument}}|
    \; &\le \; \Const^n (\randomt_\epsilon - \randomt_{\epsilon,\epsilon'})
    (\epsilon'/\delta)^{n+2} \ed^{\Const (n+1)\ln(n+2)}
    \; \le \; 
    \Const^n \randomt_\epsilon \ed^{-(n+2)\big( (1-2\eta)\ln(1/\epsilon) - \Const\ln (n+2)  \big)}
    \nonumber\\
    \; &\le \; 
    \randomt_\epsilon \ed^{- (1-3\eta)(n+2)\ln(1/\epsilon)}\label{eq: upper bound rhs}
    \end{align}
    provided $\epsilon$ is small enough (recall the choice of $\epsilon$ in \eqref{eq: defining parameter}).
    
    Let us next derive a lower bound on the absolute value of the left-hand side of \eqref{eq: main argument}.
    For this, we first derive a lower bound on the terms involving $H$, and then an upper bound on the terms involving $u$. 
    From the definitions of $\randomt_\epsilon$ and $\randomt_{\epsilon,\epsilon'}$ in \eqref{eq: t epsilon} and \eqref{eq: t epsilon epsilon prime} respectively, we find
    \[
        H_{[a-1,b]}(\randomt_{\epsilon,\epsilon'}) - H_{[a-1,b]}(\randomt_\epsilon)
        \; \ge \; \epsilon' - (|R(n,\delta;x(\randomt_{\epsilon,\epsilon'}))|+1)\epsilon
        \; \ge \; \left(\epsilon^{-\eta} - \ln \randomt_{\epsilon,\epsilon'}-1\right)\epsilon.
    \]
    We may assume that $\randomt_{\epsilon,\epsilon'}<\varphi(\epsilon)$, with $\varphi$ defined in \eqref{eq: definition of phi}),
    otherwise the result is obviously true since $\randomt_{\epsilon,\epsilon'} \le \randomt_\epsilon$. For $\epsilon$ small enough, we find 
    \[
        \epsilon^{-\eta} - \ln \randomt_{\epsilon,\epsilon'} - 1
        \; \ge \; 
        \epsilon^{-\eta} - \frac12 (\ln(1/\epsilon))^{1+\theta} - 1
        \; \ge \; 2.
    \]
    Hence
    \begin{equation}\label{eq: lower bound on energy difference}
        |H_{[a-1,b]}(\randomt_{\epsilon,\epsilon'}) - H_{[a-1,b]}(\randomt_\epsilon)|
        \; \ge \; 2 \epsilon.
    \end{equation}

    We finally need an upper bound on the terms involving $u$ in the left-hand side of \eqref{eq: main argument}.
    Proceeding as for the right-hand side of \eqref{eq: main argument}, and using now \eqref{eq: solutions u and f order n} and \eqref{eq: explicit u i} as representations for $u_{a-1}$ and $u_{b+1}$, we find 
    \begin{multline}\label{eq: upper bound terms in the lhs}
        \left| 
        \left( u_{a-1}(\randomt_\epsilon) - u_{a-1}(\randomt_{\epsilon,\epsilon'}) \right)
    - \left( u_{b+1}(\randomt_\epsilon) - u_{b+1}(\randomt_{\epsilon,\epsilon'})\right)
        \right| \\
        \; \le \; 
        \sum_{i=1}^n (\epsilon'/\delta)^{i+1} \ed^{\Const (i+1) \ln (i+1)}
        \; \le \;
        \sum_{i=1}^n \ed^{-(i+1)\big( (1 - 2\eta)\ln(1/\epsilon) - \Const \ln (n+1) \big)}
        \; \le \; \Const \epsilon^{2(1-3\eta)}
    \end{multline}
    provided $\epsilon$ is taken small enough.

    Combining \eqref{eq: lower bound on energy difference} and \eqref{eq: upper bound terms in the lhs}, we find that the left-hand side of \eqref{eq: main argument} is lower-bounded in absolute value by $\epsilon$, provided $\epsilon$ is taken small enough. 
    Combining this with \eqref{eq: upper bound rhs}, we find 
    \[
    \randomt_\epsilon 
    \;\ge\; 
    \epsilon \ed^{(1-3\eta)(n+2)\ln(1/\epsilon)}
    \;\ge\; 
    \ed^{(1-3\eta)n\ln 1/\epsilon} 
    \]
    which, using \eqref{eq: defining parameter}, yields the claim.
\end{proof}

\subsection{Proof of Theorem~\ref{th: main theorem}}

Let now let $\epsilon$ depend on time.
For $t\ge 1$, we set
\begin{equation}\label{eq: epsilon of t}
    \epsilon(t) \; = \; \ed^{-2 (\ln t)^{\frac{1}{1+\theta}}}.
\end{equation}
The function $\epsilon(t)$ is non-increasing and coincides with the threshold in the right-hand side of \eqref{eq: result main theorem} with $\theta$ as in \eqref{eq: values of eta and theta}.
Further, for $\varphi$ defined in \eqref{eq: definition of phi}, we compute that
\begin{equation}\label{eq: phi epsilon t}
    \varphi(\epsilon(t)) \; = \;  t^{2^{1+\theta}/2} \; \ge \; t+1
\end{equation}
for $t$ large enough. 
The parameters $n$, $\delta$ and $\epsilon'$ now also depend on time through \eqref{eq: defining parameter}.

We define the (possibly infinite) random time $\randomt^*$ featuring in Theorem~\ref{th: main theorem} as 
\[
    \randomt^* \; = \; \sup\{t\ge 1 : M(t) \le \epsilon(t)\}.
\]
We observe that, since $\epsilon(t)$ is decreasing, for any $t_0\ge 1$, 
\[
    \randomt^* \;\ge\; t_0
    \quad \Rightarrow \quad
    \exists t\ge t_0 \; : \; 
    \min_{0\le s \le t} M(s) \;\le\; \epsilon(t)
    \quad \Rightarrow \quad
    \exists t\ge t_0 \; : \; 
    \min_{0\le s \le \lfloor t\rfloor + 1} M(s) \;\le\; \epsilon(\lfloor t\rfloor).
\]
Hence 
\[
    \Proba(\randomt^* \ge t_0)
    \; \le \; 
    \sum_{t\in\N: t\ge \lfloor t_0 \rfloor} 
    \Proba\left(
    \min_{0\le s \le t + 1} M(s) \;\le\; \epsilon(t) 
    \right)
    \; =: \;
    \sum_{t\in\N: t\ge \lfloor t_0 \rfloor}  p(t).
\]
Therefore, to complete the proof of Theorem 1, it suffices to prove that $p(t)$ decays faster than any power law in $t$ as $t \to \infty$. Indeed, this will imply that $P(\randomt^* \ge t_0)$ also decays faster than any power law in $t_0$ as $t_0 \to \infty$, and consequently, that $\randomt^*$ has finite moments of all orders.


For this, we note that 
\[
    p(t) \; = \; \Proba\left(
    \min_{0\le s \le t + 1} M(s) \;\le\; \epsilon(t) 
    \right)
    \; = \;
    \Proba\left( \randomt_{\epsilon (t)} < t+1 \right)
\]
and we now obtain a bound on this last quantity. 
By \eqref{eq: phi epsilon t} and Lemma~\ref{lem: probabilistic improvement on the a priori bound}, we conclude that $\randomt_{\epsilon (t)} < t+1$ implies $|R(n(t),\delta(t);x(\randomt_{\epsilon(t),\epsilon'(t)}))| > \ln \randomt_{\epsilon(t),\epsilon'(t)}$
for $t$ large enough.
Proposition~\ref{pro: conclusion probabilistic bound} yields an upper bound on the probability of this event, provided Assumption~\eqref{eq: condition n and delta} and the bound $\ln(c/\epsilon'(t))\ge 1$ are satisfied. This is the case for $t$ large enough, and we find 
\begin{align*}
    \Proba(\randomt_{\epsilon} < t+1 ) 
    \;&\le\; 
    \Proba \left( 
    |R(n,\delta;x(\randomt_{\epsilon,\epsilon'}))| > \ln \randomt_{\epsilon,\epsilon'}
    \right)\\
    \; &\le \; 
    \frac{\Const}{\epsilon'} \ed^{2\ln (c/\epsilon')}
    \ed^{\big( \Const_1 (n+1) \ln (n+1) - \eta \ln(1/\epsilon) \big) \frac{\ln (c/\epsilon')}{2(n+1)}}\\
    \; &\le \; 
    \frac{\Const}{\epsilon} \ed^{-\eta \frac{\ln (1/\epsilon) \ln (c/\epsilon')}{4 (\ln(1/\epsilon))^\theta}}
\end{align*}
where the dependence on $t$ has been omitted for simplicity.
Using~\eqref{eq: epsilon of t}, we find that the exponent in this last expression is lower bounded by $c (\ln t)^{\frac{2-\theta}{1+\theta}} = c (\ln t)^{5/4}$, hence the left-hand side vanishes faster than any inverse power law in $t$. 
This proves Theorem~\ref{th: main theorem}.

\subsection{The DNLS chain}\label{sec: DNLS}
We justify here Remark~\ref{rem: dnls chain}. 
First of all, we notice that the dynamics generated by the Schrödinger equation~\eqref{eq: schrodinger equation} is Hamiltonian, and that a smooth local observable $O$ evolve as $\dd O / \dd t  = \{ H,O\}$ provided we define the Poisson bracket as 
\[
    \{f,g\} \; = \;
    \imag \left(\nabla_{\overline\psi}f \nabla_\psi g 
    - \nabla_{\psi}f \nabla_{\overline\psi} g \right).
\]
The comparison with the system studied in this paper is made evident in the variables introduced in Section~\ref{sec: change of variables}, if we set
\[
    a_x^- \; = \; \psi_x, \qquad a_x^+ \; = \; \overline \psi_x. 
\]
Doing so, we observe in particular that the rules \eqref{eq: poisson bracket rules} for the Poisson bracket, that the expressions \eqref{eq: H har} for the harmonic part of the Hamiltonian, \eqref{eq: H an expansion} for the anharmonic part, and \eqref{eq: current new variables} for the current still hold. 
This is all that is required for the proof to carry over.

\bibliographystyle{plain}
\bibliography{bibliography_slow_spreading}

\end{document}